\documentclass[10pt,journal,compsoc]{IEEEtran}
\usepackage{ifpdf}
\ifCLASSOPTIONcompsoc
  \usepackage[nocompress]{cite}
\else
  \usepackage{cite}
\fi
\usepackage[fleqn]{amsmath}
\usepackage{pgfplotstable}

\usepackage{amsthm}
\usepackage{thmtools}
\usepackage{thm-restate}

\usepackage{amssymb}
\usepackage{graphicx}
\usepackage{ragged2e}
\usepackage[font=scriptsize]{caption}
\usepackage{algorithm}
\usepackage{algpseudocode}
\usepackage{chngcntr}
\usepackage{mathtools}
\mathtoolsset{showonlyrefs}
\usepackage[english]{babel}
\newtheorem{remark}{Remark}

\newtheorem{definition}{Definition}
\newtheorem{theorem}{Theorem}

\newtheorem{assumption}{Assumption}

\usepackage[mathscr]{euscript}
\usepackage[round]{natbib}
\usepackage[colorlinks = true,
            linkcolor = blue,
            urlcolor  = blue,
            citecolor = blue,
            anchorcolor = blue]{hyperref}

\ifCLASSINFOpdf
\else
\fi
\begin{document}
%
% paper title
% Titles are generally capitalized except for words such as a, an, and, as,
% at, but, by, for, in, nor, of, on, or, the, to and up, which are usually
% not capitalized unless they are the first or last word of the title.
% Linebreaks \\ can be used within to get better formatting as desired.
% Do not put math or special symbols in the title.
\title{Convergence guarantees for response prediction \\ for latent structure network time series}
%
%
% author names and IEEE memberships
% note positions of commas and nonbreaking spaces ( ~ ) LaTeX will not break
% a structure at a ~ so this keeps an author's name from being broken across
% two lines.
% use \thanks{} to gain access to the first footnote area
% a separate \thanks must be used for each paragraph as LaTeX2e's \thanks
% was not built to handle multiple paragraphs
%
%
%\IEEEcompsocitemizethanks is a special \thanks that produces the bulleted
% lists the Computer Society journals use for "first footnote" author
% affiliations. Use \IEEEcompsocthanksitem which works much like \item
% for each affiliation group. When not in compsoc mode,
% \IEEEcompsocitemizethanks becomes like \thanks and
% \IEEEcompsocthanksitem becomes a line break with idention. This
% facilitates dual compilation, although admittedly the differences in the
% desired content of \author between the different types of papers makes a
% one-size-fits-all approach a daunting prospect. For instance, compsoc 
% journal papers have the author affiliations above the "Manuscript
% received ..."  text while in non-compsoc journals this is reversed. Sigh.

\author{
Aranyak Acharyya,
%~\IEEEmembership{Member,~IEEE,}
Francesco Sanna Passino,
%~\IEEEmembership{Fellow,~OSA,}
Michael W. Trosset,
%~\IEEEmembership{Fellow,~OSA,}and~
Carey E. Priebe
%~\IEEEmembership{Life~Fellow,~IEEE}% <-this % stops a space
\IEEEcompsocitemizethanks{\IEEEcompsocthanksitem Aranyak Acharyya is with the Mathematical Institute for Data Science, Johns Hopkins University, Baltimore,
MD, 21218. E-mail: aachary6@jhu.edu
\IEEEcompsocthanksitem Francesco Sanna Passino is with the Department of Mathematics, Imperial College London, London,
United Kingdom. E-mail: f.sannapassino@imperial.ac.uk
\IEEEcompsocthanksitem Michael W. Trosset is with the Department of Statistics, Indiana University Bloomington, Bloomington,
IN 47405. E-mail: mtrosset@iu.edu
\IEEEcompsocthanksitem Carey E. Priebe is with the Department of Applied Mathematics and Statistics, Johns Hopkins University, 
Baltimore,
MD 21218. E-mail: cep@jhu.edu
%\protect\\
% note need leading \protect in front of \\ to get a newline within \thanks as
% \\ is fragile and will error, could use \hfil\break instead.

%\IEEEcompsocthanksitem J. Doe and J. Doe are with Anonymous University.
}
% <-this % stops an unwanted space
%\thanks{Manuscript received April 19, 2005; revised August 26, 2015.}
}

\IEEEtitleabstractindextext{%
\begin{abstract}
In this article, we propose a technique to predict the response associated with an unlabeled time series of networks in a semisupervised setting. Our model involves a collection of time series of random networks of growing size, where some of the time series are associated with responses. Assuming that the collection of time series admits an unknown lower dimensional structure, our method exploits the underlying structure to consistently predict responses at the unlabeled time series of networks. Each time series represents a multilayer network on a common set of nodes, and raw stress embedding, a popular dimensionality reduction tool, is used for capturing the unknown latent low dimensional structure. Apart from establishing theoretical convergence guarantees and supporting them with numerical results, we demonstrate the use of our method in the analysis of real-world biological learning circuits of larval \textit{Drosophila}.
\end{abstract}

% Note that keywords are not normally used for peerreview papers.
}

% make the title area
\maketitle

% To allow for easy dual compilation without having to reenter the
% abstract/keywords data, the \IEEEtitleabstractindextext text will
% not be used in maketitle, but will appear (i.e., to be "transported")
% here as \IEEEdisplaynontitleabstractindextext when the compsoc 
% or transmag modes are not selected <OR> if conference mode is selected 
% - because all conference papers position the abstract like regular
% papers do.
\IEEEdisplaynontitleabstractindextext
% \IEEEdisplaynontitleabstractindextext has no effect when using
% compsoc or transmag under a non-conference mode.

% For peer review papers, you can put extra information on the cover
% page as needed:
% \ifCLASSOPTIONpeerreview
% \begin{center} \bfseries EDICS Category: 3-BBND \end{center}
% \fi
%
% For peerreview papers, this IEEEtran command inserts a page break and
% creates the second title. It will be ignored for other modes.
\IEEEpeerreviewmaketitle

\IEEEraisesectionheading{\section{Introduction}\label{sec:introduction}}
 In recent times, tools for statistical analysis and inference on random graphs have gained popularity owing to their applicability in extracting information from network data arising from various domains of real life, including neuroscience 
 \citep{Vogelstein2011GraphCU}, biology and social studies \citep{Holland1983StochasticBF}. Erd\H{o}s--R\'enyi random graphs \citep{Erdos1984OnTE} comprise the simplest model of random graphs where the probability of edge formation between any pair of nodes is equal. Stochastic blockmodels 
 \citep{Holland1983StochasticBF}
 are graphs where each node is assigned a community membership and the probability of edge formation between any pair of nodes depends only upon the corresponding community memberships. Random dot product graphs \citep{Young2007RandomDP,athreya2017statistical} represent a generalization of stochastic blockmodels, in which every node is assigned a feature vector, also known as its latent positions, and the probability of formation of an edge between any two nodes is the inner product between the corresponding latent positions. The notion of generalized random dot product graphs \citep{RubinDelanchy2017ASI} offer a further generalization to random dot product graphs, where the probability of edge formation between a given pair of nodes is the indefinite inner product between the corresponding latent positions. 

 While single random graphs have been largely explored in the recent years, the field of studying multiple networks is still emerging. In most cases, the study of multiple networks entails analysis of a multilayer network, which amounts to a collection of graphs on the same set of nodes. In \cite{jones2020multilayer}, a popular model, named the \textit{multilayer random dot product graph}, has been proposed to capture the behaviour of a multilayer graph, and a method to obtain node-level embeddings is also proposed. Results
 in \cite{gallagher2021spectral} show that \textit{unfolded adjacency spectral embedding}, the multiple network embedding procedure proposed in \cite{jones2020multilayer}, offers certain desirable stability guarantees, that is, if two nodes behave similarly then they are assigned similar embeddings up to noise.

 A popular conjecture suggests that in majority of real-life datasets, the dimension of the datapoints is only artificially high, and in essence the high-dimensional datapoints lie on or cluster around some low-dimensional manifold \citep{whiteley2022discovering}. 
This provides the motivation behind manifold learning.
Multidimensional scaling
\citep{borg2005modern}
comprise a class of procedures meant to learn the underlying low-dimensional structure that given high-dimensional datapoints correspond to.
Works in \cite{rubin2020manifold} show that the adjacency spectral embeddings of a latent position random graph with high-dimensional latent positions will be close to a low-dimensional manifold. Results in
\cite{trosset2020learning} and \cite{trosset2024continuous} establish that manifold learning can be carried out consistently from noisy datapoints sufficiently close to a low-dimensional manifold in a high-dimensional ambient space.
Based on these results, \cite{acharyya2023semisupervised} and \cite{acharyya2024consistent} respectively show that node-level and graph-level responses can be predicted in a semisupervised setting from observations on single and multiple graphs corresponding to datapoints on a low-dimensional manifold in a high-dimensional ambient space. 
This work, where every time series of networks corresponds to a high dimensional datapoint on a low dimensional manifold, can be regarded as an extension to %the previous works 
\citet{acharyya2023semisupervised}
and \citet{acharyya2024consistent}.

In this paper, our model involves a collection of time series of networks, each time series corresponding to a point on a one-dimensional manifold in a high-dimensional ambient space. Some of the time series are assumed to be associated with responses linked to the corresponding scalar pre-images via a linear regression model. We propose a technique based on the works in \cite{baum2024doubly} to predict the response at an unlabeled time series, by exploiting the presence of the scalar pre-images. We establish convergence guarantees of our algorithm and demonstrate its performance guarantees numerically. Besides, we demonstrate the use of our method to analyze the learning circuit of a collection of \textit{Drosophila} larvae. 

We organize the manuscript in the following manner. \textit{Section \ref{Sec:Definitions_and_notations}}
introduces the reader to preliminaries of topics like multiplex random graphs, \textit{Doubly Unfolded Adjacency Spectral Embedding} \citep[DUASE;][]{baum2024doubly}
and raw stress embedding \citep{borg2005modern,trosset2024continuous}. Then \textit{Section \ref{Sec:Model_and_Method}} describes our model, mentions our goal and states our proposed algorithm, and is followed by \textit{Section \ref{Sec:Theoretical_results}} which states our theoretical results. In \textit{Section \ref{Sec:Simulations}}, the numerical results are shown in support of the theoretical findings. An illustration of the use of our method in analyzing biological learning circuits of \textit{Drosophila} is presented in \textit{Section \ref{Sec:Real_Data}}, followed by a conclusion in \textit{Section \ref{Sec:Discussion}}.
The proofs of our theoretical results are given in \textit{Section \ref{Sec:Appendix}}.

% Computer Society journal (but not conference!) papers do something unusual
% with the very first section heading (almost always called "Introduction").
% They place it ABOVE the main text! IEEEtran.cls does not automatically do
% this for you, but you can achieve this effect with the provided
% \IEEEraisesectionheading{} command. Note the need to keep any \label that
% is to refer to the section immediately after \section in the above as
% \IEEEraisesectionheading puts \section within a raised box.

% The very first letter is a 2 line initial drop letter followed
% by the rest of the first word in caps (small caps for compsoc).
% 
% form to use if the first word consists of a single letter:
% \IEEEPARstart{A}{demo} file is ....
% 
% form to use if you need the single drop letter followed by
% normal text (unknown if ever used by the IEEE):
% \IEEEPARstart{A}{}demo file is ....
% 
% Some journals put the first two words in caps:
% \IEEEPARstart{T}{his demo} file is ....
% 
% Here we have the typical use of a "T" for an initial drop letter
% and "HIS" in caps to complete the first word.
%\IEEEPARstart{T}{his} demo file is intended to serve as a ``starter file''
%for IEEE Computer Society journal papers produced under \LaTeX\ using
%IEEEtran.cls version 1.8b and later.
% You must have at least 2 lines in the paragraph with the drop letter
% (should never be an issue)
%I wish you the best of success.

%\hfill mds
 
%\hfill August 26, 2015

\section{Important definitions, notations and terminologies}
\label{Sec:Definitions_and_notations}
Here, we denote the set $\lbrace
1,2,\hdots n
\rbrace$ by $[n]$. Also, for any $k,n \in \mathbb{N}$, the set $\lbrace
(k-1)n+1,(k-1)n+2,\hdots kn
\rbrace$ is denoted by $\mathscr{S}^k_n$.
In this paper, every vector will be represented by a bold lower case letter such as $\mathbf{v}$. Any vector by default is a column vector. Matrices will be denoted by bold upper case letters such as $\mathbf{A}$. For a matrix $\mathbf{A}$, the $(i,j)$-th entry will be given by $\mathbf{A}_{i,j}$, the $i$-th row (written as a column vector) will be given by $\mathbf{A}_{i,*}$
and the $j$-th column will be given by $\mathbf{A}_{*,j}$. For any matrix $\mathbf{A} \in \mathbb{R}^{m \times n}$ with $\mathrm{rank}(\mathbf{A})=r$, the singular values in descending order will be given by $\sigma_1(\mathbf{A}) \geq \dots \geq \sigma_r(\mathbf{A})$, the
corresponding left singular vectors will be given by $\mathbf{u}_1(\mathbf{A}),\dots,\mathbf{u}_r(\mathbf{A})$ and the corresponding right singular vectors will be given by $\mathbf{v}_1(\mathbf{A}),\dots,\mathbf{v}_r(\mathbf{A})$. The $n \times n$ centering matrix will be denoted by $\mathbf{H}_n=\mathbf{I}_n-
\frac{1}{n}
(\boldsymbol{1}_n \boldsymbol{1}_n^T)$ where $\mathbf{I}_n$ is the $n \times n$ identity matrix and $\boldsymbol{1}_n$ is the $n$-dimensional vector of all ones. For a matrix $\mathbf{A} \in \mathbb{R}^{m \times n}$, $\mathbf{A}_{[\mathscr{S}_1,\mathscr{S}_2]}$ (where $\mathscr{S}_1 \subseteq [m]$, $\mathscr{S}_2 \subseteq [n]$) denotes the matrix obtained by retaining the rows with indices in $\mathscr{S}_1$ and the columns with indices in $\mathscr{S}_2$, and 
$\mathbf{A}_{[\mathscr{S}_1,.]}$ denotes the matrix obtained by retaining the  rows with indices in $\mathscr{S}_1$ and all the columns of $\mathbf{A}$, and $\mathbf{A}_{[.,\mathscr{S}_2]}$ denotes the matrix obtained by retaining the  columns with indices in $\mathscr{S}_2$ and all the rows of $\mathbf{A}$.

Discussed below are some important definitions and notions that we will frequently utilize in this paper.

\subsection{Preliminiaries on
 multilayer graphs and DUASE
}
\label{Subsec:Prelims_SBM_COSIE}

A graph is an ordered pair $(V,E)$ where $V$ denotes the set of vertices and $E \subseteq V \times V$ denotes the collection of edges. An adjacency matrix $\mathbf{A}$ of a graph is defined in the following manner:
$\mathbf{A}_{i,j}=1$ if $(i,j) \in E$, and $\mathbf{A}_{i,j}=0$ otherwise. Here, we deal with directed graphs, hence $\mathbf{A}$ has a positive probability of being asymmetric. Latent position random graphs are those graphs where each node is associated with a vector that is called its latent position. The latent position of the $i$-th node is denoted by $\mathbf{x}_i \in \mathbb{R}^d$ for some natural number $d$. 
First, we state the definition of random graphs.
\begin{definition}[Random graph; \citet{Holland1983StochasticBF}]
\label{Def:RdirG}
    Suppose $G$ is a directed random graph with $n$ nodes, such that the probability of an edge from the $i$-th node to the $j$-th node is given by $p_{ij}$. Then, the probability matrix of outward edges for the graph $G$ will be given by $\mathbf{P}=\left( p_{ij} \right)_{i,j=1}^n$, henceforth referred to as outward edge formation probability matrix (and sometimes we will drop `outward' for sake of convenience). The adjacency matrix $\mathbf{A} \in \mathbb{R}^{n \times n}$ satisfies $\mathbf{A}_{i,j} \sim^{ind} \mathrm{Bernoulli}(p_{ij})$, for all $i,j\in[n],\ i \neq j$, and $\mathbf{A}_{i,i}=0,\ i\in[n]$.
\end{definition}
 Secondly, we state the formal definition of multilayer directed random graphs.
\begin{definition}[Multilayer random graph; \citet{jones2020multilayer}]
\label{Def:multilayer_random_graphs}
     A multilayer random graph is a collection of graphs with a common set of nodes, but varying probability of edge from one node to another, for any pair of nodes. Suppose 
     $G_1,\dots,G_M$ denote a multilayer graph with a common set $V$ of $n$ nodes and $M$ layers. The probability of an edge from the $i$-th node to the $j$-th node is given by $p^{(l)}_{ij}$ for the graph $G_l$, for all $l \in [M]$. The adjacency matrix $\mathbf{A}^{(l)} \in \mathbb{R}^{n \times n}$ of $G_l$ satisfies $\mathbf{A}^{(l)}_{i,j} \sim^{ind} \mathrm{Bernoulli}(p^{(l)}_{ij})$, for all $i,j\in[n],\ i \neq j,\ l \in [M]$, and $\mathbf{A}^{(l)}_{i,i}=0,\ i\in[n],\ l \in [M]$.
 \end{definition}
 
In this paper, we will deal with a collection of multilayer random graphs, %and our method is based on the method of Doubly
%Unfolded
%Adjacency Spectral Embedding \citep{baum2024doubly} that can reliably extract information from a set of multilayer graphs. 
%\begin{definition}
 %   In a collection of multilayer random graphs, 
 where the number $n$ of nodes remains the same across all observations.  In total, there are $N$  multilayers of networks and there are $M$ graphs in each multilayer. In our notation, $G^{(k,l)}$ denotes the $l$-th network in the $k$-th layer,
    and the corresponding adjacency matrix is given by $\mathbf{A}^{(k,l)} \in \mathbb{R}^{n \times n}$
    for all $k \in [N], l \in [M]$. 
    In our paper, we shall be dealing with a collection of $N$ time series of graphs (in which each series contains $M$ graphs on a common set of $n$ nodes) where each time series can be regarded as a multilayer graph. For each $l \in [M]$,
    the $l$-th graph $G^{(k,l)}$ corresponds to a timepoint $\tau_l \in [0,\tau^*],\ \tau^*\in\mathbb R_+$, for all $k \in [N]$,
    for some values 
    $0<\tau_1<\tau_2<\dots \tau_{N-1}<\tau_N<\tau^*$. %, $\tau^* \in \mathbb{R}_+$.
%\end{definition}
For the kind of collection of multilayer graphs defined above, a reliable method of embedding for subsequent inference is \textit{Doubly Unfolded Adjacency Spectral Embedding} \citep[DUASE;][]{baum2024doubly}. It offers the stability guarantee that if two nodes (with the possibility that they belong to different graphs in different layers) behave similarly, then they will be assigned similar embedding.  
The algorithmic pseudocode for DUASE is described in Algorithm~\ref{Algo:DUASE}. 

\begin{algorithm}[t]
\caption{DUASE$\Big(
\left\lbrace
\mathbf{A}^{(k,l)}
\right\rbrace_{k \in [N],l \in [M]}
;d
\Big)$} 
\label{Algo:DUASE}
\begin{algorithmic}[1]
\State Construct the block matrix $$\boldsymbol{\mathcal{A}}=
\left(
\mathbf{A}^{(k,l)}
\right)_{k \in [N],l \in [M]}
\in \mathbb{R}^{nN \times nM}
.$$
\State Define 
$\mathbf{U}_{\boldsymbol{\mathcal{A}}}=[\mathbf{u}_1(\boldsymbol{\mathcal{A}})| \dots \mathbf{u}_d(\boldsymbol{\mathcal{A}}) ] \in \mathbb{R}^{nN \times d}$ to be the matrix of the top $d$ left singular vectors, $\mathbf{V}_{\boldsymbol{\mathcal{A}}}=[\mathbf{v}_1(\boldsymbol{\mathcal{A}})| \dots \mathbf{v}_d(\boldsymbol{\mathcal{A}}) ] \in \mathbb{R}^{nM \times d}$ to be the matrix of the top $d$ right singular vectors, and $\boldsymbol{\Sigma}_{\boldsymbol{\mathcal{A}}}=\mathrm{diag}\{
\sigma_1(\boldsymbol{\mathcal{A}}),\dots
\sigma_d(\boldsymbol{\mathcal{A}})\}
\in \mathbb{R}^{d \times d}$ to be the diagonal matrix of the top $d$ singular values of $\boldsymbol{\mathcal{A}}$.
\State Compute the left embedding 
$\mathbf{X}_{\boldsymbol{\mathcal{A}}}=\mathbf{U}_{\boldsymbol{\mathcal{A}}} \boldsymbol{\Sigma}_{\boldsymbol{\mathcal{A}}}^{\frac{1}{2}}
\in \mathbb{R}^{nN \times d}
$ and
the right embedding 
$\mathbf{Y}_{\boldsymbol{\mathcal{A}}}=\mathbf{V}_{\boldsymbol{\mathcal{A}}} \boldsymbol{\Sigma}_{\boldsymbol{\mathcal{A}}}^{\frac{1}{2}}
\in \mathbb{R}^{nM \times d}
$.
\State The information of the $k$-th layer across all timepoints is contained in $\mathbf{X}^{(k)}_{\boldsymbol{\mathcal{A}}}=
\left(
\mathbf{X}_{\boldsymbol{\mathcal{A}}}
\right)_{[(k-1)n+1:kn,.]} \in \mathbb{R}^{n \times d}
$ for all $k \in [N]$ and the information of the $l$-th timepoint across all layers is contained in 
$\mathbf{Y}^{(l)}_{\boldsymbol{\mathcal{A}}}=
\left(
\mathbf{Y}_{\boldsymbol{\mathcal{A}}}
\right)_{[(l-1)n+1:ln,.]}
\in \mathbb{R}^{n \times d}
$ for all $l \in [M]$.
\State \Return 
$\lbrace
\mathbf{X}^{(k)}_{\boldsymbol{\mathcal{A}}}: k \in [N]
\rbrace$.
\end{algorithmic}
\end{algorithm}

\begin{remark}
    In the original paper on \textit{Doubly Unfolded Adjacency Spectral Embedding} \citep{baum2024doubly}, the DUASE algorithm technically returns both the left embeddings
    $
    \lbrace 
\mathbf{X}^{(k)}_{\boldsymbol{\mathcal{A}}}: k \in [N]
    \rbrace
    $ and the right embeddings 
    $
    \lbrace 
\mathbf{Y}^{(l)}_{\boldsymbol{\mathcal{A}}}: l \in [M]
    \rbrace
    $. However, since our goal in this paper needs only the left embeddings, we state the algorithm as in Algorithm~\ref{Algo:DUASE}, returning only $\lbrace
\mathbf{X}^{(k)}_{\boldsymbol{\mathcal{A}}}: k \in [N]
\rbrace$.
\end{remark}

%\newline
\subsection{Raw stress embedding} 
\label{Subsec:Raw_stress_embedding}
\subsubsection{Finite sample size}
Raw stress embedding is a popular method for nonlinear dimensionality reduction. Given dissimilarities $\left\lbrace 
\boldsymbol{\Delta}_{i,j}
\right\rbrace_{i,j \in [N]}$ for some finite $N \in \mathbb{N}$, the goal is to find vectors 
$\mathbf{z}_1,\dots, \mathbf{z}_N \in  \mathbb{R}^c$ for some predetermined target embedding dimension $c$, such that the interpoint Euclidean distances between the $\mathbf{z}_i$ approximate the corresponding dissimilarities, that is, $\left\lVert 
\mathbf{z}_i-\mathbf{z}_j
\right\rVert \approx \boldsymbol{\Delta}_{i,j}$ for all $i,j \in [N]$. 
The full algorithm is given below.
\begin{algorithm}[H]
\caption{RSEmb(
$\left\lbrace 
\boldsymbol{\Delta}_{i,j}
\right\rbrace_{i,j \in [N]};c)$}
\label{Algo:raw_stress_finite_sample}
\begin{algorithmic}[1]
\State Define the raw stress function to be 
$
\sigma(\mathbf{z}_1,\mathbf{z}_2,\dots ,\mathbf{z}_N)=
\sum_{i,j=1}^N 
w_{i,j}
(
\left\lVert \mathbf{z}_i-\mathbf{z}_j \right\rVert-
\boldsymbol{\Delta}_{i,j}
)^2
$ where $\mathbf{z}_i \in \mathbb{R}^c$ for all $i \in [N]$.
%\State 
\State Obtain
$(\hat{\mathbf{z}}_1,....,\hat{\mathbf{z}}_N)
= \arg \min \sigma(\mathbf{z}_1,\mathbf{z}_2, \dots ,\mathbf{z}_N)$
\State \Return $(\hat{\mathbf{z}}_1,\dots ,\hat{\mathbf{z}}_N)$.
\end{algorithmic}
\end{algorithm}
For this article, we restrict our attention to the regime of $c=1$. Moreover, we set $w_{i,j}=1$ for all pairs $(i,j)$.
\begin{remark}
     Iterative majorization \citep[for details, see Chapter $8$ of][]{borg2005modern} is used to minimize the raw stress function. In order to avoid getting trapped in local minima, classical multidimensional scaling outputs are typically used for initialization. In our paper, we assume the global minima is reached for theoretical reasons.
\end{remark}
\subsubsection{Infinite sample size}
\label{Subsubsec:Taw_Stress_infinite_sample}
Suppose $\mathcal{M}$ is a compact Riemannian manifold of innate dimension $c$, and let $\boldsymbol{\Delta}: \mathcal{M} \times \mathcal{M} \to \mathbb{R}$ be a Borel-measurable dissimilarity function. Assuming $g:\mathcal{M} \to \mathbb{R}^c$ is a Borel-measurable embedding function,
define the raw stress function by 
\begin{equation*}
\begin{aligned}
&\sigma((\boldsymbol{\Delta},\mathscr{P}),g) \\ &=
    \int_{\mathcal{M}}
    \int_{\mathcal{M}}
    \left(
    \left\lVert 
     g(m')-g(m'')
    \right\rVert -
    \boldsymbol{\Delta}(m',m'')
    \right)^2 
    \mathscr{P}(dm')
    \mathscr{P}(dm'').
\end{aligned}
\end{equation*}
Defining $\mathbf{D}(m',m'')=
\left\lVert 
g(m')-g(m'')
\right\rVert
$, we redefine the raw stress function as 
\begin{equation*}
\begin{aligned}    &\sigma((\boldsymbol{\Delta},\mathscr{P}),\mathbf{D}) \\ &=
    \int_{\mathcal{M}}
    \int_{\mathcal{M}}
    \left(
    \left\lVert 
     g(m')-g(m'')
    \right\rVert -
    \boldsymbol{\Delta}(m',m'')
    \right)^2 
    \mathscr{P}(dm')
    \mathscr{P}(dm'') \\ &=
    \left\lVert 
    \mathbf{D}-\boldsymbol{\Delta}
    \right\rVert^2_{\mathscr{P}}.
\end{aligned}
\end{equation*}
Let $\mathscr{Y}$ be the cone of all Euclidean pseudometrics, then we define the raw stress minimizer for dissimilarity $\boldsymbol{\Delta}$ with respect to probability function $\mathscr{P}$ as 
\begin{equation*}
    \mathrm{Min}(\boldsymbol{\Delta},\mathscr{P})=
    \left\lbrace 
     \mathbf{D} \in \mathscr{Y}:
     \sigma((\boldsymbol{\Delta},\mathscr{P}),\mathbf{D})
     \leq 
     \inf_{\mathbf{D} \in \mathscr{Y}}
     \sigma((\boldsymbol{\Delta},\mathscr{P}),\mathbf{D})
    \right\rbrace.
\end{equation*}

\section{Model and Methodology} 
\label{Sec:Model_and_Method}
Our model involves a set of time series of graphs. Each time series can be regarded as a multilayer random directed graph. This is so because a multilayer random graph is a collection of graphs on a common set of nodes, while a time series of graphs in practice involves a collection of realizations of a single graph over multiple timepoints (for instance, sequence of snapshots of a network of neurons in the brain of an organism). It is assumed that some of the time series are associated with a scalar response. It is also assumed that each time series corresponds to a scalar pre-image. A simple linear regression model links the responses to the scalar pre-images of the time series of graphs. 

There are $N$ time series in total, and each time series has $M$ graphs, where each graph has $n$ nodes. The adjacency matrix of the $l$-th graph in the $k$-th time series is denoted by $\mathbf{A}^{(k,l)}$, and let the corresponding probability matrix be $\mathbf{P}^{(k,l)}$.
We define the grand probability matrix to be $\boldsymbol{\mathcal{P}}=
(
\mathbf{P}^{(k,l)}
)_{k \in [N],l \in [M]}
$.
We denote $$
\lbrace 
\mathbf{X}^{(k)}_{\boldsymbol{\mathcal{P}}}: k \in [N]
\rbrace
=\mathrm{DUASE}\left(\left\lbrace 
\mathbf{P}^{(k,l)}
\right\rbrace_{k \in [N],l \in [M]}\right),$$ where we recall that for every $k \in [N]$, $\mathbf{X}^{(k)}_{\boldsymbol{\mathcal{P}}} \in \mathbb{R}^{n \times d}$. Suppose there exist scalars $t_i$ 
such that for all $k_1,k_2$,
\begin{equation*}
    \left\lVert 
\mathbf{X}^{(k_1)}_{\boldsymbol{\mathcal{P}}}-
\mathbf{X}^{(k_2)}_{\boldsymbol{\mathcal{P}}}
\right\rVert_{2,\infty}=
|t_{k_1}-t_{k_2}|.
\end{equation*}
Suppose $s \ll K$ is a fixed natural number and for $k \in [s]$, response $y_k$ is associated with the $k$-th time series. We further assume that the responses $y_k$ are linked to the scalar pre-images $t_k$ via a simple linear regression model, that is,
\begin{equation*}
    y_k=\alpha+\beta t_k+\epsilon_k
\end{equation*}
where $\epsilon_k \sim^{iid} N(0,\sigma^2_{\epsilon})$. 

Our goal is to predict $y_{s+1}$ for the unlabeled $(s+1)$-th time series. To do that, we first estimate $\mathbf{X}^{(k)}_{\boldsymbol{\mathcal{P}}}$ with $\mathbf{X}^{(k)}_{\boldsymbol{\mathcal{A}}}$, where 
$$
\lbrace 
\mathbf{X}^{(k)}_{\boldsymbol{\mathcal{A}}}: k \in [N]
\rbrace
=\mathrm{DUASE}\left(\left\lbrace 
\mathbf{A}^{(k,l)}
\right\rbrace_{k \in [N],l \in [M]}\right).$$ The matrices $\mathbf{X}^{(k)}_{\boldsymbol{\mathcal{A}}}$ estimate their population counterpart matrices
$\mathbf{X}^{(k)}_{\boldsymbol{\mathcal{P}}}$ consistently up to an orthogonal transformation, and hence the pairwise distances between
$ 
\mathbf{X}^{(k)}_{\boldsymbol{\mathcal{A}}}
$
can consistently estimate the corresponding pairwise distances 
$
\mathbf{X}^{(k)}_{\boldsymbol{\mathcal{P}}}
$. Hence the population (involving probability matrices $\mathbf{P}^{(k,l)}$) dissimilarity matrix $\boldsymbol{\Delta}$ is estimated by the sample (involving adjacency matrices $\mathbf{A}^{(k,l)}$) dissimilarity matrix $\hat{\boldsymbol{\Delta}}$, where 
\begin{align} 
&\boldsymbol{\Delta}=
    \left(
\left\lVert
\mathbf{X}^{(k_1)}_{\boldsymbol{\mathcal{P}}}-
\mathbf{X}^{(k_2)}_{\boldsymbol{\mathcal{P}}}
\right\rVert
    \right)_{k_1,k_2=1}^N
=
\left(
|t_{k_1}-t_{k_2}|
\right)_{k_1,k_2=1}^N, \notag
    \\
&\hat{\boldsymbol{\Delta}}=
  \left(
\left\lVert
\mathbf{X}^{(k_1)}_{\boldsymbol{\mathcal{A}}}-
\mathbf{X}^{(k_2)}_{\boldsymbol{\mathcal{A}}}
\right\rVert
    \right)_{k_1,k_2=1}^N.
\end{align}
Then, we apply raw stress minimization algorithm on $\hat{\boldsymbol{\Delta}}$ into $\mathbb{R}$, and obtain scalars 
$\left\lbrace
\hat{z}_k
\right\rbrace_{k=1}^N
$. Treating the embeddings $\hat{z}_k$ as proxy regressors, we predict the response $y_{s+1}$ with $\Tilde{y}_{s+1}$. The full procedure is described in \textit{Algorithm~\ref{Algo:response_TSG_prediction}}.

\begin{algorithm}[H]
\caption{PredTSGResp$\Big(
\left\lbrace
\mathbf{A}^{(k,l)}
\right\rbrace_{k \in [N],l \in [M]}
;d;r
\Big)$} 
\label{Algo:response_TSG_prediction}
\begin{algorithmic}[1]
\State Compute the DUASE of $\lbrace
\mathbf{A}^{(k,l)}
\rbrace_{k \in [N],l \in [M]}$:
$$
\left\lbrace 
\mathbf{X}^{(k)}_{\boldsymbol{\mathcal{A}}}
\right\rbrace_{k=1}^N
=
\mathrm{DUASE}
\left(\left\lbrace 
\mathbf{A}^{(k,l)}
\right\rbrace_{k \in [N],l \in [M]}\right).
$$
\State Obtain the estimated pairwise distance matrix
$$\hat{\boldsymbol{\Delta}}=
\left(
\left\lVert 
\mathbf{X}^{(k_1)}_{\boldsymbol{\mathcal{A}}}-
\mathbf{X}^{(k_2)}_{\boldsymbol{\mathcal{A}}}
\right\rVert_{2,\infty}
\right)_{k_1,k_2=1}^N
.$$
\State Use raw stress embedding to obtain 
$$
\left\lbrace 
\hat{z}_k
\right\rbrace_{k=1}^{N}=
\mathrm{RSEmb}(\hat{\boldsymbol{\Delta}};1)
.$$
\State Compute the estimated regression coefficients:
$$\hat{b}=
\frac
{
\sum_{i=1}^s (y_i-\Bar{y})(\hat{z}_i-\Bar{\hat{z}})
}
{
\sum_{i=1}^s (\hat{z}_i-\Bar{\hat{z}})^2
}, \quad
\hat{a}=\Bar{y}-\hat{b}\Bar{\hat{z}},
$$
where $\Bar{\hat{z}}=s^{-1}\sum_{i=1}^s\hat z_i$.
\State \Return Predicted responses
$\Tilde{y}_{r}=\hat{a}+\hat{b}\hat{z}_{r}$.
\end{algorithmic}
\end{algorithm}

From now onwards, we shall be dealing with scenarios where $N,M,n$ vary all together. Specifically, there exists a sequence $(N_K,M_K,n_K)$ such that 
 $N_K \to \infty,\ M_K \to \infty,\ n_K \to \infty$ as $K \to \infty$ in a manner that the \textit{Theorem \ref{Th:DUASE_XA_const_est_XP}} holds. If it is clear from the context that the $K$-th instant is being spoken of, we shall omit the subscript $K$ and replace $N_K$ with $N$ (and $M_K$ with $M$, $n_K$ with $n$).
 Stated below are our model assumptions. 
 \begin{assumption}
    \label{Asm:rank_of_grand_P_matrix}
    The grand probability matrix, defined as $\boldsymbol{\mathcal{P}}=(\mathbf{P}^{(k,l)})_{k \in [N_K], l \in [M_K]}$, satisfies $\mathrm{rank}(\boldsymbol{\mathcal{P}})=d$ for all sufficiently large $K$.
 \end{assumption}
 \begin{assumption}
     \label{Asm:rates_of_growth}
     The number of graphs per multilayer and the number of multilayers both must grow slower than the number of nodes in each graph, that is, $M_K=o(n_K)$, $N_K=o(n_K)$.
 \end{assumption}
 \begin{assumption}
\label{Asm:scalars_in_compact_set}
The scalars $t_k \sim^{iid} \mathscr{P}$ for all $k$, where $\mathcal{L}=\mathrm{support}(\mathscr{P}) \subset \mathbb{R}$ is closed and bounded. 
 \end{assumption}
 Having discussed our setting in this section, we move on to the next section to state the results establishing asymptotic properties of our proposed PredTSGResp algorithm. % (described in \textit{Algorithm \ref{Algo:response_TSG_prediction}}).
 
 \section{Theoretical results}
\label{Sec:Theoretical_results}
In this section, we present our theoretical results. The results are primarily based on two results from the literature, one that establishes consistency of DUASE embeddings (\textit{Theorem~\ref{Th:DUASE_XA_const_est_XP}}) and another that establishes continuity of raw stress embeddings (\textit{Theorem~\ref{Th:dissimilarity_minimizer_convergence}}). 
\begin{theorem}
\citep{baum2024doubly}
\label{Th:DUASE_XA_const_est_XP}
    Suppose there are $N$ time series of networks, each consisting of $M$ directed random latent position graphs where each graph has $n$ nodes. Each time series can be regarded as a multiplex of graphs. Let 
    $\mathbf{P}^{(k,l)}$ denote the outward edge formation probability matrix and
    $\mathbf{A}^{(k,l)}$ denote the adjacency matrix of the $l$-th graph in the $k$-th series, for all $k \in [N], l \in [M]$. Further, assume 
    that \textit{Assumptions \ref{Asm:rank_of_grand_P_matrix}} and \textit{\ref{Asm:rates_of_growth}} hold.
    Denoting 
    $$\lbrace 
\mathbf{X}^{(k)}_{\boldsymbol{\mathcal{P}}}: k \in [N]
    \rbrace=
    \mathrm{DUASE}\left(
    \left\lbrace 
    \mathbf{P}^{(k,l)}
    \right\rbrace_{k \in [N],l \in [M]};d
    \right)
    $$ and
    $$\lbrace 
\mathbf{X}^{(k)}_{\boldsymbol{\mathcal{A}}}: k \in [N]
    \rbrace=
    \mathrm{DUASE}\left(
    \left\lbrace 
    \mathbf{A}^{(k,l)}
    \right\rbrace_{k \in [N],l \in [M]};d
    \right),
    $$ 
    there exists $\mathbf{Q} \in \mathcal{O}(d)$ such that
    for each $k \in \mathbb{N}$,
    \begin{equation*}
        \left\lVert 
        \mathbf{X}^{(k)}_{\boldsymbol{\mathcal{A}}}-\mathbf{X}^{(k)}_{\boldsymbol{\mathcal{P}}}
        \mathbf{Q}
        \right\rVert_{2,\infty} \to^P 0
    \end{equation*}
as $K \to \infty$.
\end{theorem}
The above theorem states that as the number of time series, number of graphs and number of nodes increase simultaneously, the sample DUASE embedding for every time series converges to its population counterpart up to an orthogonal transformation. Since the maximum row norm of a matrix is invariant to orthogonal transformations, we can conclude that the distance between two sample DUASE embeddings approach the distance between the two corresponding population DUASE embeddings. 
\begin{restatable}[]{proposition}{propone}
\label{Prop:DUASE_pairwise_max_diff_zero}
    In the setting of \textit{Theorem \ref{Th:DUASE_XA_const_est_XP}}, define the matrices 
\begin{equation}
\begin{aligned}
\boldsymbol{\Delta}^{(K)}=
    \left(
    \left\lVert 
    \mathbf{X}^{(k_1)}_{\boldsymbol{\mathcal{P}}}-
    \mathbf{X}^{(k_2)}_{\boldsymbol{\mathcal{P}}}
    \right\rVert_{2,\infty}
    \right)_{k_1,k_2 \in [N]}, \\
    %\text{and} \\
\hat{\boldsymbol{\Delta}}^{(K)}=
    \left(
    \left\lVert 
    \mathbf{X}^{(k_1)}_{\boldsymbol{\mathcal{A}}}-
    \mathbf{X}^{(k_2)}_{\boldsymbol{\mathcal{A}}}
    \right\rVert_{2,\infty}
    \right)_{k_1,k_2 \in [N]}.
\end{aligned}
\end{equation}
     Then, for each $k_1,k_2 \in \mathbb{N}$,
    \begin{equation*}
        \left|
        \hat{\boldsymbol{\Delta}}^{(K)}_{k_1,k_2}-
        \boldsymbol{\Delta}^{(K)}_{k_1,k_2}
        \right| \to^P 0
    \end{equation*}
as $K \to \infty$.
\end{restatable}
%\newline
%\newline
%\newline
Observe that as the number of time series, number of graphs and number of nodes grow simultaneously (under our model assumptions, one of which demands that the number of multilayers and number of graphs per multilayer must grow slower than the the number of nodes per graph), each entry of the population pairwise distance matrix $\boldsymbol{\Delta}$ converges entrywise to the pairwise distance matrix $\mathbf{D}$ between the scalar pre-images $t_i$. Our next result from \cite{trosset2024continuous}
states that the globally minimizing EDM-1 matrices of a sequence of dissimilarity matrices converge to the globally minimizing EDM-1 matrix of the limit of the dissimilarity matrices. 
\begin{theorem}
\label{Th:dissimilarity_minimizer_convergence}
\citep{trosset2024continuous}
    Suppose $\hat{\boldsymbol{\Delta}}^{(K)}$
    is uniformly bounded and converges to the dissimilarity function $\hat{\boldsymbol{\Delta}}^{(\infty)}$
    in the topology of pointwise convergence. Assume that the sequence of empirical probability measures $\lbrace \hat{\mathscr{P}}_K \rbrace$ weakly converges to the probability measure $\mathscr{P}_{\infty}$. Then every sequence $\lbrace 
    \hat{\mathbf{D}}^{(K)}
    \rbrace$ where $\hat{\mathbf{D}}^{(K)} \in \mathrm{Min}(\hat{\boldsymbol{\Delta}}^{(K)},\hat{\mathscr{P}}_{K})$, will have an accumulation point, and if $\hat{\mathbf{D}}^{(\infty)}$
    is an accumulation point of $\lbrace
    \hat{\mathbf{D}}^{(K)}
    \rbrace$, then it will satisfy
    $\hat{\mathbf{D}}^{(\infty)} \in \mathrm{Min}(\hat{\boldsymbol{\Delta}}^{(\infty)},\mathscr{P}_{\infty})$.
\end{theorem}
The above theorem establishes consistency of dissimilarity minimizers in a growing sample size scenario. It helps us prove the following result which establishes the consistency of raw stress embeddings obtained from the maximum row norm differences between the sample DUASE embeddings. 

\begin{restatable}[]{proposition}{proptwo}
\label{Prop:Raw_stress_embedding_consistency}
    In the setting of \textit{Theorem \ref{Th:DUASE_XA_const_est_XP}},    
define $$\left\lbrace 
\hat{z}_i
\right\rbrace_{i=1}^{N}=
\mathrm{RSEmb}
(
\hat{\boldsymbol{\Delta}}^{(K)};1
)
.$$ Then as $K \to \infty$, for every $k_1,k_2 \in \mathbb{N}$, 
$$
\left(
|\hat{z}_{k_1}-\hat{z}_{k_2}|-
|t_{k_1}-t_{k_2}|
\right) 
\to^P 0.
$$
\end{restatable}
%\newline
%\newline
%\newline
%\newline
Recall from our assumptions that the pre-images $t_k$ are the unknown regressors in our regression model for which we wish to predict a response. \textit{Proposition \ref{Prop:Raw_stress_embedding_consistency}} tells us that pairwise distances between the raw stress embeddings approach the pairwise distance between the true regressors, thereby helping the raw stress embeddings closely approximate an affine transformation of the regressors, justifying the use of the raw stress embeddings as proxy regressors in a linear regression model.

\begin{restatable}[]{theorem}{theothree}
\label{Th:predicted_response_consistency}
Suppose there are $N$ time series of random directed graphs where each series has $M$ graphs and each graph has $n$ nodes. Define 
$\boldsymbol{\mathcal{P}}=
(
\mathbf{P}^{(k,l)}
)_{k \in [N],l \in [M]}
$ and $\boldsymbol{\mathcal{A}}=
(
\mathbf{A}^{(k,l)}
)_{k \in [N],l \in [M]}
$, where $\mathbf{P}^{(k,l)}$ and $\mathbf{A}^{(k,l)}$ denote the probability matrix and the adjacency matrix for the $l$-th graph in the $k$-th series. 
Denote the population DUASE embeddings by 
$
\lbrace 
\mathbf{X}^{(k)}_{\boldsymbol{\mathcal{P}}}: k \in [N]
\rbrace=
\mathrm{DUASE}(
\mathbf{P}^{(k,l)};d
)
$ and the sample DUASE embeddings by 
$
\lbrace 
\mathbf{X}^{(k)}_{\boldsymbol{\mathcal{A}}}: k \in [N]
\rbrace=
\mathrm{DUASE}(
\mathbf{A}^{(k,l)};d
)
$ and assume that  there exist scalar pre-images $t_k$ such that for every $k_1,k_2$, 
\begin{equation}
    \left\lVert 
\mathbf{X}^{(k_1)}_{\boldsymbol{\mathcal{P}}}-
\mathbf{X}^{(k_2)}_{\boldsymbol{\mathcal{P}}}
\right\rVert_{2,\infty}=
|t_{k_1}-t_{k_2}|.
\end{equation}
 Suppose responses $y_1,\dots y_s$ are observed corresponding to the first $s$ time series such that the following model hold
\begin{equation*}
    y_k=\alpha+\beta t_k+\epsilon_k
\end{equation*}
where $\epsilon_k \sim^{iid} N(0,\sigma^2_{\epsilon}), k \in [s]$. Denote the raw stress embeddings by $\left\lbrace \hat{z}_k \right\rbrace_{k=1}^N=
\mathrm{RSEmb}(\hat{\boldsymbol{\Delta}};1)
$ where
\begin{equation*}
    \hat{\boldsymbol{\Delta}}=
\left(
\left\lVert
\mathbf{X}^{(k_1)}_{\boldsymbol{\mathcal{A}}}-
\mathbf{X}^{(k_2)}_{\boldsymbol{\mathcal{A}}}
\right\rVert_{2,\infty}
\right)_{k_1,k_2=1}^N.
\end{equation*}
Then the predicted response 
$$\Tilde{y}_r=\mathrm{PredTSGResp}\left(
\left\lbrace 
\mathbf{A}^{(k,l)}
\right\rbrace_{k \in [N],l \in [M]};d,r
\right)$$ satisfies $$|\Tilde{y}_{r}-\hat{y}_{r}| \to^P 0$$ as $K \to \infty$, where $\hat{y}_{r}$ is the predicted response at the $r$-th time series based on the true regressors $t_k$.
\end{restatable}

The above theorem tells us that as the number of unlabeled (auxilary) time series increases (along with number of graphs and size of graphs), the predicted response obtained from our method approaches the predicted response obtained from the true regressors. In order to test the validity of a simple linear regression model, we deploy an $F$-test that uses the observed responses and the predicted responses obtained from the true regressors. In the absence of the true regressors, we can still use predicted responses obtained from our method and hope to mimic the power of the original $F$-test, by virtue of \textit{Theorem \ref{Th:predicted_response_consistency}}.

\begin{restatable}[]{corollary}{coroone}
\label{Cor:power_convergence}
    In the setting of \textit{Theorem \ref{Th:predicted_response_consistency}}, suppose we want to test $H_0:\beta=0$ against $H_1:\beta \neq 0$ at level of significance $\Tilde{\alpha}$. Define the following test statistics:
    \begin{equation*}
        F^*=(s-2)
        \frac
        {
        \sum_{k=1}^s (\hat{y}_k-\Bar{y})^2
        }
        {
        \sum_{k=1}^s (y_k-\hat{y}_k)^2
        },
        \hspace{0.1cm}
        \hat{F}=(s-2)
        \frac
        {
        \sum_{k=1}^s (\tilde{y}_k-\Bar{y})^2
        }
        {
        \sum_{k=1}^s (y_k-\tilde{y}_k)^2
        }.
    \end{equation*}
Suppose $\pi^*$ is the power of the test carried out by the principle: reject $H_0$ if $F^*>c_{\Tilde{\alpha}}$, and let $\hat{\pi}$ be the power of the test with the principle: reject $H_0$ if $\hat{F}>c_{\Tilde{\alpha}}$. Then, for every $(\alpha,\beta)$, 
$|\hat{\pi}-\pi^*| \to 0$ as $K \to \infty$.
\end{restatable}

The above result paves a way for testing the validity of a proposed linear regression model between the responses and the scalar pre-images in a realistic setting where the regressors are unknown.

\section{Simulations}
\label{Sec:Simulations}
In this section, we describe the simulation experiments. We carry out two simulation experiments that provide numerical support for the theoretical results \textit{Theorem~\ref{Th:predicted_response_consistency}} and \textit{Corollary~\ref{Cor:power_convergence}}. 

First, we describe the simulation that numerically shows the predicted response obtained from \textit{Algorithm \ref{Algo:response_TSG_prediction}} approaches the predicted response obtained from the true regressors. 
The number of labeled datapoints is fixed at $s=5$. Denoting the common index that controls the growth of $n$ (number of nodes per graph), $M$ (number of graphs per multilayer) and $N$ (number of multilayers) by $K$,
we set $n_K=15+\lfloor (K-1)^{1.5} \rfloor$, $N_K=10+(K-1)$ and $M_K=8+(K-1)$, while the
 common index $K$ varies in the range $\left\lbrace
1,2,3,\dots 30
\right\rbrace$. 
We repeat the following task for each $K$ on each of $100$ Monte Carlo samples.
At first, we obtain the regressors  $t_1,\dots t_s \sim^{iid} U(0,1)$
associated with observed responses
$y_1,\dots y_s$ where $y_k=\alpha+\beta t_k +\epsilon_k$,
$\epsilon_k \sim^{iid} N(0,\sigma_{\epsilon}^2)$ with regression parameters $\alpha=2.0$, $\beta=8.0$, $\sigma_{\epsilon}=0.01$.
We generate the pre-images for the auxiliary points as $t_{s+1},\dots t_{N} \sim^{iid} U(0,1)$.
We define two matrices $\Tilde{\mathbf{X}} \in \mathbb{R}^{nN \times d}$ and $\Tilde{\mathbf{Y}} \in \mathbb{R}^{nM \times d}$, where we choose the embedding dimension $d=2$.
Now,
\begin{equation*}
    \Tilde{\mathbf{X}}_{[((i-1)n+1):in,]}=
    \frac{t_i}{\sqrt{d}} \mathbf{J}_{(n,d)}
%\begin{pmatrix}
%    \frac{t_i}{\sqrt{d}} & \frac{t_i}{\sqrt{d}} & \dots & \frac{t_i}{\sqrt{d}} \\
%    \frac{t_i}{\sqrt{d}} & \frac{t_i}{\sqrt{d}} & \dots & \frac{t_i}{\sqrt{d}} \\
%    \dots \\
%    \dots \\
%    \frac{t_i}{\sqrt{d}} & \frac{t_i}{\sqrt{d}} & \dots & \frac{t_i}{\sqrt{d}}
%\end{pmatrix}
\end{equation*}
where $\mathbf{J}_{(n,d)}$ is the $n \times d$ matrix of all ones.
 Also,
for every pair $(i,j) \in [nM] \times [d]$, $\Tilde{\mathbf{Y}}_{ij} \sim^{iid} U(0.2,0.5)$. We define the grand probability matrix $\boldsymbol{\mathcal{P}}=\Tilde{\mathbf{X}} \Tilde{\mathbf{Y}}^T$. 
For any $i_1,i_2 \in [nN]$, 
$$\left\lVert  \Tilde{\mathbf{X}}_{i_1,*}-\Tilde{\mathbf{X}}_{i_2,*}
\right\rVert=
\left\lVert
\left(
\mathbf{U}_{\boldsymbol{\mathcal{P}}} \mathbf{S}_{\boldsymbol{\mathcal{P}}}^{\frac{1}{2}+\eta}
\right)_{i_1,*}-
\left(
\mathbf{U}_{\boldsymbol{\mathcal{P}}} \mathbf{S}_{\boldsymbol{\mathcal{P}}}^{\frac{1}{2}+\eta}
\right)_{i_2,*}
\right\rVert
.$$
Hence, for any $k_1,k_2 \in [N]$,
\begin{equation*}
\begin{aligned} 
&|t_{k_1}-t_{k_2}| \\
&=\left\lVert 
\Tilde{\mathbf{X}}_{[\mathscr{S}^{k_1}_n,]}-
\Tilde{\mathbf{X}}_{[\mathscr{S}^{k_2}_n,]}
\right\rVert_{2,\infty} \\
&=
\left\lVert
\left(
\mathbf{U}_{\boldsymbol{\mathcal{P}}} \mathbf{S}_{\boldsymbol{\mathcal{P}}}^{\frac{1}{2}+\eta}
\right)_{[\mathscr{S}^{k_1}_n,]}-
\left(
\mathbf{U}_{\boldsymbol{\mathcal{P}}} \mathbf{S}_{\boldsymbol{\mathcal{P}}}^{\frac{1}{2}+\eta}
\right)_{[\mathscr{S}^{k_2}_n,]}
\right\rVert_{2,\infty} \\
&=
\left\lVert
\left\lbrace
\left(
\mathbf{U}_{\boldsymbol{\mathcal{P}}}
\mathbf{S}_{\boldsymbol{\mathcal{P}}}^{\frac{1}{2}}
\right)_{[\mathscr{S}^{k_1}_n,]}
-
\left(
\mathbf{U}_{\boldsymbol{\mathcal{P}}}
\mathbf{S}_{\boldsymbol{\mathcal{P}}}^{\frac{1}{2}}
\right)_{[
\mathscr{S}^{k_2}_n,
]}
\right\rbrace 
\mathbf{S}_{\boldsymbol{\mathcal{P}}}^{\eta}
\right\rVert_{2,\infty}
\end{aligned}
\end{equation*} 
where $\eta \in [-\frac{1}{2},\frac{1}{2}]$. Clearly, if $d=1$, then we could scale $t_k$ by
${\sigma_1(\boldsymbol{\mathcal{P}})^{\eta}}$ for all $k$, and reformulated our regression model. But even if $d>1$, it can be numerically shown that the second largest singular value of $\boldsymbol{\mathcal{P}}$ is much smaller than its largest singular value, and hence the scaling of the regressors and the reformulation of the regression model still holds approximately, enough to ensure numerical results are satisfactory. We provide numerical evidence in \textit{Appendix~\ref{sec:numerical_justification}}.

Next, we generate the adjacency matrix $\mathbf{A}^{(k,l)}$ for the $l$-th graph in the $k$-th multilayer following $\mathbf{A}^{(k,l)}_{i,j} \sim^{ind} \mathrm{Bernoulli}(\mathbf{P}^{(k,l)}_{i,j})$ for all $i \neq j$, where $\mathbf{P}^{(k,l)}$ is the probability matrix for the $l$-th graph in the $k$-th multilayer. We obtain the DUASE embedding 
$$
\lbrace
\mathbf{X}^{(k)}_{\boldsymbol{\mathcal{A}}}: k \in [N]
\rbrace
=
\mathrm{DUASE}\left(
\left\lbrace
\mathbf{A}^{(k,l)}
\right\rbrace_{k \in [N],l \in [M]};
d\right)
$$ and thereby compute the raw stress embeddings 
$$
\left\lbrace \hat{z}_i \right\rbrace_{i=1}^N=
\mathrm{RSEmb}(
\lbrace
\hat{\boldsymbol{\Delta}}_{i,j}
\rbrace_{i,j \in [N]}
;1
)
.$$
Using a simple linear regression model on $(y_i,t_i)_{i=1}^s$, the response at the $(s+1)$-th multilayer is predicted with $\hat{y}_{s+1}$.
Similarly, using a simple linear regression model on $(y_i,\hat{z}_i)_{i=1}^s$, we predict the response at the $(s+1)$-th multilayer by $\Tilde{y}_{s+1}$. The mean of the values of the squared difference $(\hat{y}_{s+1}-\Tilde{y}_{s+1})^2$ is computed over all the $100$ Monte Carlo samples, and plotted against $K$, and the resulting plot is given in \textit{Figure \ref{fig:pred_res_const}}. It is seen that the sample average squared distance between the predicted response from the true regressors and the predicted response obtained from our method approaches zero as $K$ goes to infinity, thus supporting \textit{Theorem \ref{Th:predicted_response_consistency}}.

\begin{figure}[t]
    \centering
\includegraphics[scale=0.65]{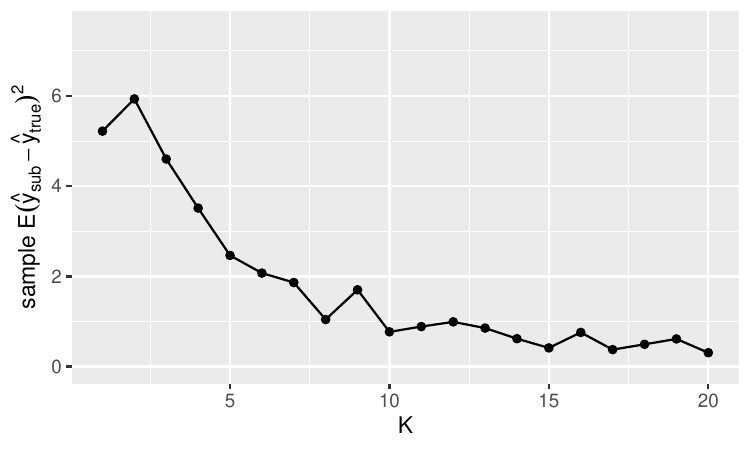}
    \caption{Plot showing that average squared difference between the predicted response based on the true regressors and the predicted response obtained from our method approaches zero, as the number of multilayers, number of graphs per multilayer and the number of nodes per graph increases in a suitable manner.}
    \label{fig:pred_res_const}
\end{figure}

We next present simulation results in support of \textit{Corollary~\ref{Cor:power_convergence}}. 
For testing $H_0:\beta=0$ versus $H_1:\beta \neq 0$, we choose level of significance $\Tilde{\alpha}$. 
The setting is same as before, except this time we take $n_K=15+ \lfloor (K-1)^{1.5} \rfloor$, 
$M_K=8+(K-1)$ and $N_K=10+(K-1)$. For every $K$ in the range $\left\lbrace 1,2,\dots 30 \right\rbrace$, on each of $100$ Monte Carlo samples, we proceed as before to obtain the raw-stress embeddings $\left\lbrace 
\hat{z}_i
\right\rbrace_{i=1}^N$. Using a linear regression model on the bivariate data $(y_i,t_i)_{i=1}^s$ we compute the predicted responses $\left\lbrace \hat{y}_i \right\rbrace_{i=1}^s$ and thus obtain the true $F$-statistic $F^*$, and 
 using a linear regression model on the bivariate data
$(y_i,\hat{z}_i)_{i=1}^s$, we obtain the predicted responses $\left\lbrace \Tilde{y}_i \right\rbrace_{i=1}^s$  and thus obtain the substitute $F$-statistic
$\hat{F}$, where
\begin{equation*}
    F^*=(s-2)
        \frac
        {
        \sum_{k=1}^s (\hat{y}_k-\Bar{y})^2
        }
        {
        \sum_{k=1}^s (y_k-\hat{y}_k)^2
        }, \hspace{0.15cm}
\hat{F}=(s-2)
        \frac
        {
        \sum_{k=1}^s (\tilde{y}_k-\Bar{y})^2
        }
        {
        \sum_{k=1}^s (y_k-\tilde{y}_k)^2
        }.
\end{equation*}
The test based on $F^*$ (and equivalently, also the test based on $\hat{F}$) rejects $H_0$ at significance level $\Tilde{\alpha}$ if $F^*>c_{\Tilde{\alpha}}$ for pre-specified threshold $c_{\Tilde{\alpha}}$.
For each statistic amongst $F^*$ and $\hat{F}$, we estimate the 
power of the test based on that statistic by computing the
proportion of times the test based on that statistic rejects $H_0$ at level $\Tilde{\alpha}$. We calculate the absolute difference between the estimated powers of the tests based on $F^*$ and $\hat{F}$ and plot them against $K$, and the resulting plot is given in \textit{Figure \ref{fig:power_conv}}.  We observe that the difference between the estimated powers of the tests approaches zero as $K$ increases. 
\begin{figure}[t]
    \centering
\includegraphics[scale=0.65]{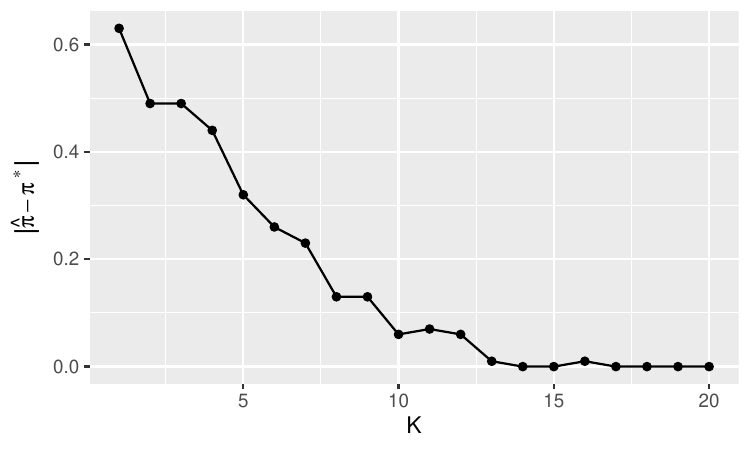}
    \caption{Plot showing the difference between estimated powers of the tests based on the true and the substitute $F$-statistics approaching zero as the number of multilayers, number of graphs per multilayer and number of nodes per graph increase at suitable rates. A set of $100$ Monte Carlo samples of a collection of multilayer directed random graphs are generated, and one-dimensional embeddings are obtained by raw-stress minimization on the Double Unfolded Adjacency Spectral Embeddings. Responses are regressed against these one-dimensional raw-stress embeddings to obtain a substitute $F$-statistic.}
    \label{fig:power_conv}
\end{figure}
\section{Real Data Analysis}
\label{Sec:Real_Data}
In this section, we demonstrate the use of our methodology in the analysis of biological learning networks of larval \textit{Drosophila}. The wiring diagram, also known as connectome, of the larval \textit{Drosophila} has been recently completed \citep{winding2023connectome}, which enables simulation of biologically realistic models of the circuits of neurons based on known anatomical connectivity \citep{eschbach2020recurrent}. There have been recent works on studying the learning networks (circuit of neurons responsible for learning in an organism)  by training connectome-constrained models to perform associative learning in simualtions where a given stimulus is delivered, eliciting a certain network output in the animal (for instance, when an odour is coupled with pain, the odour loses its attraction to the organism).

To be more specific, in our case, the network models are trained to perform extinction learning. In this phenomenon, an association between a conditioned stimulus (e.g an odour) and reinforcement (e.g pain) is initially learnt, and then weakened by exposure to the same conditioned stimulus in absence of the reinforcement. The behaviour of the network over different timepoints is simulated, comprising a time series of networks corresponding to a single extinction learning trial. A total of $143$ such trials are performed corresponding to $11$ different replications for each of $13$ different models, each model being a result of removal of a particular synapse from the parent network. In each trial, at first a conditioned stimulus, followed by reinforcement (pain or reward), is delivered, and then after significant gaps the stimulus is delivered again twice, without being coupled with the reinforcement.
A learning score is recorded for every extinction learning trial. The learning score is defined as the ratio of the network output at the third conditioned stimulus to that at the second conditioned stimulus, where the network output at a particular time is defined as the ratio of degree of aversion to the degree of attraction to the conditioned stimulus at that time. 

We thus have $M=143$ time series of networks, each associated with a learning score. Each time series has $N=160$ networks and each network has $n=140$ nodes. We convert each network into a binary one by choosing to record the entry of the adjacency matrix as one if its modulus exceeds a particular threshold, and as zero otherwise. The threshold is taken to be the $25$-th percentile of the absolute values of the original edge weights. We perform Double Unfolded Adjacency Spectral Embedding on the collection of these time series of networks (with embedding dimension $d=3$), and thus obtain a matrix representation for every time series. We obtain a dissimilarity matrix of the pairwise differences (measured in two-to-infinity norm) of the matrix representations of the time series, and by subsequent raw-stress minimization we obtain one-dimensional embeddings for all the time series. A linear regression model is assumed to link the responses with the one-dimensional raw-stress embeddings, and an $F$-test with $p\text{-value}=0.026$ justifies it (at level of significance $0.05$).
The scatterplot of the responses against the one-dimensional raw-stress embeddings is givenin \textit{Figure~\ref{fig:real_pred_res_vs_rs_emb}}, along with the fitted regression line.

\begin{figure}[t]
    \centering
\includegraphics[scale=0.60]{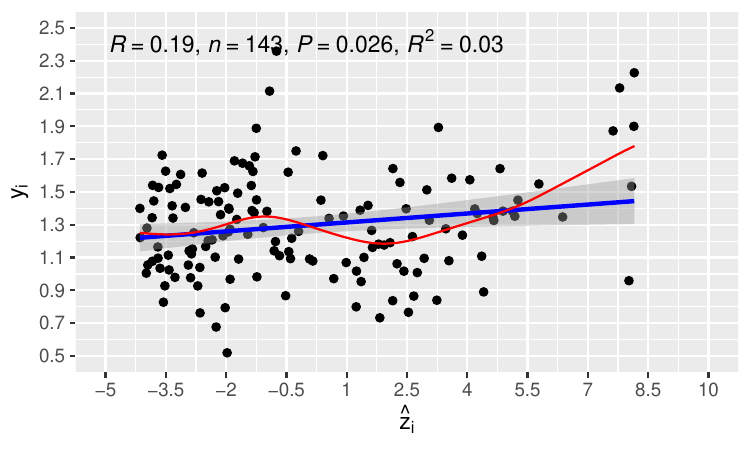}
    \caption{Scatterplot of the responses $y_i$ against the one-dimensional raw-stress embeddings $\hat{z}_i$, along with the fitted linear regression line in blue. An $F$-test is performed to check the validity of the simple regression model $y_i=a+b\hat{z}_i+\eta_i$ and $p=0.026$ is obtained, justifying the use of the simple linear regression model. The estimated model parameters are $\hat{a}=1.296$ and $\hat{b}=0.018$. The curve in red denotes the fitted nonparametric regression curve, by the method of local linear regression, which gives an R-squared value of $0.1667$. 
    }
    \label{fig:real_pred_res_vs_rs_emb}
\end{figure}

\section{Discussion}
\label{Sec:Discussion}

In this paper, we propose a method to predict a response corresponding to an unlabeled time series of networks, in a semisupervised setting. We assume that each time series of networks correspond to a scalar pre-image such that a suitable measure of pairwise distances between the time series is captured by the interpoint distances between the scalar pre-images. Assuming a linear regression model links the responses to the scalar pre-images, we propose to predict the responses by using raw-stress minimization to find proxies for the true regressors. We provide theoretical and numerical justification for our method in \textit{Sections \ref{Sec:Theoretical_results}} and \textit{\ref{Sec:Simulations}} respectively. 

We also demonstrate the use of our method in analysis of learning networks in larval \textit{Drosophila}. A collection of time series of networks, each representing the behaviour of the learning circuit in larval \textit{Drosophila} over snapshots of time in a learning trial, is observed. Each time series is associated with a response dubbed learning score. Our method obtains one-dimensional embeddings corresponding to all the time series, such that the responses can be linked to the embeddings via a linear regression model at level $0.05$ (an $F$-test results in a $p$-value of $0.026$).

The desirable asymptotic properties of our method are reliant on the guarantee of vanishing uniform bound on the errors in estimating the regressors. Such guarantee can help extend our result to the regime where a nonparametric regression is used to link the responses with the regressors. We provide an example of predicting the responses via a nonparametric regression model, from our real dataset (details in \textit{Section \ref{Sec:Real_Data}}). The setting is same as 
the one described in \textit{Section \ref{Sec:Real_Data}}, where we have $143$ time series of networks, each associated with a response. A one-dimensional embedding is obtained corresponding to every time series by raw stress minimization upon the DUASE embeddings, and a local linear regression model is used to predict the responses. The plot of the responses against the regressors, along with the fitted regression curve, is given in \textit{Figure \ref{fig:real_pred_res_vs_rs_emb}}.
%\begin{figure}
%    \centering
%\includegraphics[scale=0.60]{P3res_vs_embedding_np.pdf}
    %\caption{Scatterplot of the learning scores as responses and the one-dimensional embeddings as regressors, along with the fitted nonparametric regression curve. The R-squared value is $0.1667$.}
%\label{fig:response_vs_embeddings_nonparametric}
%\end{figure}

The results in this paper involve establishing asymptotic convergence guarantees for the output of the proposed algorithm. Finding the rate of convergence of the predicted response obtained from our method to the predicted response obtained from the true regressors is an interesting open problem in this area. Solving that problem would give us an idea of how large the set of auxiliary time series needs to be in order to achieve a given level of accuracy. Moreover, making finite-sample improvements to the algorithmic output taking the measurement error into account, comprises a potentially intriguing problem solving which will likely be beneficial to the practitioners.

% use section* for acknowledgment
\ifCLASSOPTIONcompsoc
  % The Computer Society usually uses the plural form
  \section*{Acknowledgments}
\else
  % regular IEEE prefers the singular form
  \section*{Acknowledgment}
\fi

Francesco Sanna Passino acknowledges funding from the Engineering and Physical Sciences Research Council (EPSRC), grant number EP/Y002113/1.

\bibliographystyle{plainnat}
\bibliography{ref}

\appendices

% you can choose not to have a title for an appendix
% if you want by leaving the argument blank
\section{Background justifications and proofs}
\label{Sec:Appendix}

\subsection{Certain background justifications}

\textbf{Use of Theorem 2:}
From the entrywise convergence of a sequence of dissimilarity matrices of growing size, we use \textit{Theorem \ref{Th:dissimilarity_minimizer_convergence}} to conclude the convergence of the corresponding globally minimizng EDM-1 matrices. However, \textit{Theorem \ref{Th:dissimilarity_minimizer_convergence}} actually rests on the a setting of a  sequence of dissimilarity functions 
$\lbrace 
\hat{\boldsymbol{\Delta}}^{(K)}
\rbrace$
converging pointwise to another dissimilarity function $\lbrace
\hat{\boldsymbol{\Delta}}^{(\infty)}
\rbrace$, and a sequence of probability distributions 
$\lbrace 
\hat{\mathcal{P}}_K
\rbrace$ converging uniformly 
to a probability distribution $\mathcal{P}_{\infty}$, and it states that any sequence 
$\hat{\mathbf{D}}^{(K)} \in \mathrm{Min}(\hat{\boldsymbol{\Delta}}^{(K)},\hat{\mathcal{P}}_K)$ will have an accumulation point 
$\hat{\mathbf{D}}^{(\infty)} \in \mathrm{Min}(\hat{\boldsymbol{\Delta}}^{(\infty)},\mathcal{P}_{\infty})$.
\newline
\newline
\textbf{Statement of Proposition 2:}
While the conclusion in our paper for \textit{Proposition \ref{Prop:Raw_stress_embedding_consistency}} is pointwise convergence of the raw stress embeddings, a much stronger statement holds true in this regard. The mode of convergence is $L^p$, that is:
\begin{equation*}
  \lim_{K \to \infty}  \int_{\mathcal{L}} \int_{\mathcal{L}}
  \left(
  |\hat{z}_{k_1}-\hat{z}_{k_2}|-|t_{k_1}-t_{k_2}|
  \right)^p
  \mathscr{P}(d t_{k_1})
  \mathscr{P} (d t_{k_2})
  =0
\end{equation*}
for some $p \geq 1$, under \textit{Assumption \ref{Asm:scalars_in_compact_set}}.

\subsection{Numerical justification for approximation in higher dimensional ambient spaces} \label{sec:numerical_justification}
In our simulation in \textit{Section \ref{Sec:Simulations}}, the relation
\begin{equation}
|t_{k_1}-t_{k_2}|=
\left\lVert 
\left(
\mathbf{U}_{\boldsymbol{\mathcal{P}}}
\mathbf{S}_{\boldsymbol{\mathcal{P}}}^{\frac{1}{2}}
\right)_{[\mathscr{S}^{k_1}_n,]}
-
\left(
\mathbf{U}_{\boldsymbol{\mathcal{P}}}
\mathbf{S}_{\boldsymbol{\mathcal{P}}}^{\frac{1}{2}}
\right)_{[\mathscr{S}^{k_2}_n,]}
\right\rVert_{2,\infty}
\end{equation}
up to rescaling of all the $t_k$ holds exactly for $d=1$ and approximately for $d>1$, because for $d>1$ the second largest singular value of $\boldsymbol{\mathcal{P}}$ is much smaller than its largest singular value. Numerical evidence is shown in \textit{Table~\ref{tab:num_just}}. 

\begin{table}[t]
    \centering
\begin{filecontents*}{P3num_app.txt}
1   38.3913207802614   3.79012067821168e-17
2   46.2467967723536   4.30423371659228e-18
3   53.3483912419807   2.19173601047544e-17
4   70.6204284932543   7.59644834333919e-17
5   87.3256963731673   1.60271756976842e-16
6   105.690843867424   3.75449557243395e-17
7   127.893833191334   5.12974328627644e-16
8   153.452174685367   7.32395180182678e-16
9   181.817642929274   1.08072901300687e-15
10   216.55689743623   2.64045904981925e-15
11   252.293666389214   4.88745113265974e-15
12   295.071278399014   2.1915113364941e-14
13   343.504129675349   4.74920333185641e-14
14   393.134326882093   1.0073712770101e-13
15   458.2653037287   1.97145542975649e-13
16   519.272524528707   2.97958596948405e-13
17   575.884298576391   4.07037061348058e-13
18   650.49260937855   5.36758516418727e-13
19   724.990035514813   6.43326958379969e-13
20   811.462566958402   7.98294314677235e-13
\end{filecontents*}
\pgfplotstabletypeset[
 columns/0/.style={column name={$K$}},
    columns/1/.style={column name={$\sigma_1(\boldsymbol{\mathcal{P}})$}},
    columns/2/.style={column name={$\sigma_2(\boldsymbol{\mathcal{P}})$}},
    header=false,
    string type,
    before row=\hline,
    every last row/.style={after row=\hline},
    column type/.add={|}{},
    every last column/.style={column type/.add={}{|}}
]{P3num_app.txt}
\caption{Largest and second largest singular values of the grapnd probability matrix $\boldsymbol{\mathcal{P}}$
as the number of nodes in each graph, number of graphs in each time series and the number of time series grow simultaneously.
} 
\label{tab:num_just}
\end{table}

\subsection{Proofs of theoretical results}
\label{Subsec:Proofs}

%% Restate theorem/proposition automatically, without the need of copy/paste
\propone*

\begin{proof}
Fix $k_1,k_2 \in \mathbb{N}$. Observe that by triangle inequality and invariance of Frobenius norm to orthogonal transformation,
\begin{equation*}
\begin{aligned}  
&\left\lVert 
    \mathbf{X}^{(k_1)}_{\boldsymbol{\mathcal{A}}}-
    \mathbf{X}^{(k_2)}_{\boldsymbol{\mathcal{A}}}
    \right\rVert_{2,\infty} -
    \left\lVert 
    \mathbf{X}^{(k_1)}_{\boldsymbol{\mathcal{P}}}-
    \mathbf{X}^{(k_2)}_{\boldsymbol{\mathcal{P}}}
    \right\rVert_{2,\infty} \\ 
    &\leq 
    \left\lVert 
\mathbf{X}^{(k_1)}_{\boldsymbol{\mathcal{A}}}-
\mathbf{X}^{(k_1)}_{\boldsymbol{\mathcal{P}}} \mathbf{Q}
    \right\rVert_{2,\infty} +
    \left\lVert 
\mathbf{X}^{(k_2)}_{\boldsymbol{\mathcal{A}}}-
\mathbf{X}^{(k_2)}_{\boldsymbol{\mathcal{P}}} \mathbf{Q}
    \right\rVert_{2,\infty}.
\end{aligned}
\end{equation*}
Likewise, we can write, 
\begin{equation*}
    \begin{aligned} 
&\left\lVert 
\mathbf{X}^{(k_1)}_{\boldsymbol{\mathcal{P}}}-
\mathbf{X}^{(k_2)}_{\boldsymbol{\mathcal{P}}}
\right\rVert_{2,\infty} -
\left\lVert 
\mathbf{X}^{(k_1)}_{\boldsymbol{\mathcal{A}}}-
\mathbf{X}^{(k_2)}_{\boldsymbol{\mathcal{A}}} 
\right\rVert_{2,\infty} \\ &\leq 
\left\lVert 
        \left(
\mathbf{X}^{(k_1)}_{\boldsymbol{\mathcal{A}}}-
\mathbf{X}^{(k_1)}_{\boldsymbol{\mathcal{P}}}\mathbf{Q}
        \right) -
        \left(
\mathbf{X}^{(k_2)}_{\boldsymbol{\mathcal{A}}}-
\mathbf{X}^{(k_2)}_{\boldsymbol{\mathcal{P}}}\mathbf{Q}
        \right)
        \right\rVert_{2,\infty}.
    \end{aligned}
\end{equation*}
Thus, from the above two equations combined, 
\begin{equation*}
\begin{aligned}
&\left|
   \left\lVert 
\mathbf{X}^{(k_1)}_{\boldsymbol{\mathcal{A}}}-
\mathbf{X}^{(k_2)}_{\boldsymbol{\mathcal{A}}}
\right\rVert_{2,\infty} -
\left\lVert 
\mathbf{X}^{(k_1)}_{\boldsymbol{\mathcal{P}}}-
\mathbf{X}^{(k_2)}_{\boldsymbol{\mathcal{P}}} 
\right\rVert_{2,\infty}
\right| \\
&\leq 
\left\lVert 
        \left(
\mathbf{X}^{(k_1)}_{\boldsymbol{\mathcal{A}}}-
\mathbf{X}^{(k_1)}_{\boldsymbol{\mathcal{P}}}\mathbf{Q}
        \right) -
        \left(
\mathbf{X}^{(k_2)}_{\boldsymbol{\mathcal{A}}}-
\mathbf{X}^{(k_2)}_{\boldsymbol{\mathcal{P}}}\mathbf{Q}
        \right)
        \right\rVert_{2,\infty}. 
\end{aligned}
\end{equation*}
Recall that from \textit{Theorem \ref{Th:DUASE_XA_const_est_XP}}, the right hand side goes to zero, hence so does the left hand side. Thus, for each $k_1,k_2 \in \mathbb{N}$,
\begin{equation*}
    \left|
    \hat{\boldsymbol{\Delta}}^{(K)}_{k_1,k_2}-
    \boldsymbol{\Delta}^{(K)}_{k_1,k_2}
    \right| \to^P 0
\end{equation*}
as $K \to \infty$.
\end{proof}

%% Restate Proposition 2 automatically
\proptwo*
\begin{proof}
From \textit{Proposition \ref{Prop:DUASE_pairwise_max_diff_zero}}, for  every $k_1,k_2$,
$\left|
\hat{\boldsymbol{\Delta}}_{k_1,k_2}-\boldsymbol{\Delta}_{k_1,k_2}
\right|
\to^P
0$ as $K \to \infty$. From 
model assumptions, for every $k_1,k_2$, 
$
\boldsymbol{\Delta}_{k_1,k_2}=
|t_{k_1}-t_{k_2}|
$ as $K \to \infty$. Hence, for every $k_1,k_2$,
$\left|
\hat{\boldsymbol{\Delta}}_{k_1,k_2}-|t_{k_1}-t_{k_2}| 
\right|
\to^P 0
$ as $K \to \infty$. By \textit{Theorem 3} in \cite{trosset2024continuous}, we have that 
$\left|
|\hat{z}_{k_1}-\hat{z}_{k_2}|-|t_{k_1}-t_{k_2}|
\right| \to^P 0$.
\end{proof}

%% Restate Theorem 3 automatically
\theothree*
\begin{proof}
From \textit{Proposition \ref{Prop:Raw_stress_embedding_consistency}} we have,
\begin{equation*}
    \max_{k_1,k_2 \in [s]}
    \left(
    |\hat{z}_i-\hat{z}_j|-|t_i-t_j|
    \right) \to^P 0
\end{equation*}
as $N,M,n \to \infty$. As the difference between the interpoint distances between the embeddings $\left\lbrace \hat{z}_k \right\rbrace_{k=1}^s$ and the interpoint distances between the true regressors $\left\lbrace t_k \right\rbrace_{k=1}^s$ approach zero, the raw stress embeddings approach an affine transformation on the true regressors. Since we know that an affine transformation upon the true regressors in a simple linear regression model does not alter a predicted response value, the predicted response $\Tilde{y}_r$ based on the embeddings $\hat{z}_k$ approach the predicted response $\hat{y}_r$ based on the true regressors $t_k$.
\end{proof}

%% Restate Corollary 1
\coroone*
\begin{proof}
We know from \textit{Theorem \ref{Th:predicted_response_consistency}}, for any $(\alpha,\beta) \in \mathbb{R}^2$, for all $r \in [s]$, $|\Tilde{y}_r-\hat{y}_r| \to 0$ as $K \to \infty$. Hence, for all $(\alpha,\beta) \in \mathbb{R}^2$, $|\hat{F}-F^*| \to 0$ as $K \to \infty$, and hence for all $(\alpha,\beta) \in \mathbb{R}^2$, for any significance level $\Tilde{\alpha}$, $|\hat{\pi}-\pi^*|=|\mathbb{P}_{\alpha,\beta}[\hat{F}>c_{\Tilde{\alpha}}]-\mathbb{P}_{\alpha,\beta}[F^*>c_{\Tilde{\alpha}}]| \to 0$ as $K \to \infty$.
\end{proof}

%\newpage

% Can use something like this to put references on a page
% by themselves when using endfloat and the captionsoff option.
\ifCLASSOPTIONcaptionsoff
  \newpage
\fi

\begin{IEEEbiographynophoto}{Aranyak Acharyya}
Aranyak Acharyya received his Bachelors degree in Statistics from Presidency University at Kolkata, India in 2017, and his Masters degree in Statistics from Indian Institute of Technology Kanpur in 2019. From 2019 to 2024, he was a PhD student in the Department of Applied Mathematics and Statistics at Johns Hopkins University, and then he began working as a postdoctoral fellow in the Mathematical Institute for Data Science in Johns Hopkins University. His research interests include statistical inference on networks, manifold learning and artificial intelligence.  
\end{IEEEbiographynophoto}

% if you will not have a photo at all:
\begin{IEEEbiographynophoto}{Francesco Sanna Passino}
Francesco Sanna Passino received a double BSc degree in Statistics from the universities of Bologna and Glasgow in 2016, a MSc degree in Statistics from Imperial College London in 2017, and a PhD in Statistics from Imperial College London in 2021. Since 2022, he has been a Lecturer (Assistant Professor) in Statistics at Imperial College London. His research interests mostly revolve around statistical analysis of dynamic networks, latent variable models and model-based clustering. 
\end{IEEEbiographynophoto}

% insert where needed to balance the two columns on the last page with
% biographies
%\newpage

\begin{IEEEbiographynophoto}{Michael W. Trosset}
Michael W. Trosset received his BA degree in Mathematics from Rice University 1978, and the PhD degree in Statistics from the University of California at Berkeley in 1983. He has held faculty positions at the University of Arizona, the College of William and Mary, and Indiana University where he currently chairs the Department of Statistics. He is the author of an introductory textbook, An Introduction to Statistical Inference and Its Applications to R. His research interests include euclidean representations of proximity data, nonlinear dimension reduction, computer experiments and stochastic optimization.  
\end{IEEEbiographynophoto}

\begin{IEEEbiographynophoto}{Carey E. Priebe}
Carey E. Priebe received the Bachelors degree in Mathematics from Purdue University in 1984, the Masters degree in Computer Science from San Diego State University in 1988, and the PhD degree in Information Technology from George Mason University in 1993. Since 1994 he has been a professor at the Department of Applied Mathematics and Statistics at Johns Hopkins University. His research interests include statistical analysis of random networks, statistical pattern recgnition and statistical inference on high dimensional data. He is a senior member of the IEEE, and Elected Member of the International Statistical Institute, a Fellow of the Institute of Mathematical Statistics, and a Fellow of the American Statistical Association.
\end{IEEEbiographynophoto}
\vfill

% You can push biographies down or up by placing
% a \vfill before or after them. The appropriate
% use of \vfill depends on what kind of text is
% on the last page and whether or not the columns
% are being equalized.

%\vfill

% Can be used to pull up biographies so that the bottom of the last one
% is flush with the other column.
%\enlargethispage{-5in}

% that's all folks
\end{document}